\documentclass[a4paper,twoside,11pt,fullpage]{llncs}

\usepackage{amsmath}
\usepackage{amsfonts}
\usepackage{fullpage}
\usepackage{epsfig}
\usepackage{array}
\usepackage{epsfig}
\usepackage{amsmath}
\usepackage{amssymb}
\usepackage{graphicx}
\usepackage{subfigure}
\usepackage{color}
\usepackage{algorithm}
\usepackage{algorithmic}

\begin{document}

\title{Scalable and Secure Aggregation in Distributed Networks}

\author{S\'ebastien Gambs$^1$\and Rachid Guerraoui$^2$\and Hamza Harkous$^2$\\ Florian Huc$^2$ \and Anne-Marie Kermarrec$^3$\\
$^1$Universit\'e de Rennes 1 -- INRIA/IRISA\\
$^2$\'Ecole Polytechnique F\'ed\'erale de Lausanne (EPFL)\\
$^3$INRIA Rennes Bretagne-Atlantique\\
} 
\institute{}
\maketitle

\begin{abstract}
We consider the problem of computing an aggregation function in a \emph{secure} and \emph{scalable} way. Whereas previous distributed solutions with similar security guarantees  have a communication cost of $O(n^3)$, we present a distributed protocol that requires only a communication complexity of $O(n\log^3 n)$, which we prove is near-optimal. Our protocol ensures perfect security against a computationally-bounded adversary, tolerates $(\frac{1}{2}-\epsilon)n$ malicious nodes for any constant $\frac{1}{2} > \epsilon > 0$ (not depending on $n$), and outputs the exact value of the aggregated function with high probability.
\end{abstract}

\section{Introduction}

\emph{Aggregation functions} are a specific class of functions that
compute a global function from the local values held by nodes in a
distributed system.  Examples of aggregation functions include simple
functions such as the average computation but also more sophisticated ones
such as an election with multiple candidates.  Such functions are
particularly important in large-scale systems in which they are typically used to
compute global system properties  (e.g. for monitoring purposes).
While such computations may be achieved through a trusted central entity gathering
all inputs \cite{BFP+01}, distributed variants are appealing for
scalability and privacy reasons.

In this paper, we address simultaneously three issues related to computing
an aggregation function in a distributed fashion, namely
correctness, privacy and scalability.  Indeed, the problem of
computing an aggregation function is a
specific instance of the much broader problem of secure multiparty
computation. Generic constructions
can thus be used \cite{Goldreich1987,Ben-Or1988} for solving the problem
while tolerating up to $n/2-1$ malicious nodes, where $n$ is the number of
nodes in the network. However, these constructions are often expensive
with a global communication cost that is quadratic or cubic in the
number of nodes in the network \cite{BeerliovaTrubiniova2008},
In addition, most of them assume the existence of a broadcast
 channel, which is rarely available in large scale
 networks. Simulating such a channel is possible but has a
 communication cost of $\Omega(n^2)$.

Ideally, we would like the communication complexity to compute an
aggregation function to be linear in the number of nodes. The
fundamental question is whether an algorithm for computing an
aggregation function with a communication cost of $O(n)$ in a secure
and accurate way in  the presence of a constant fraction of malicious nodes
exists. Clearly, this is  impossible to achieve with
certainty. Indeed, to be certain that a message sent by a node is not
altered by a collusion of malicious nodes, it needs to be sent at
least $m+1$ times, where $m$ is the number of malicious nodes
(otherwise the malicious nodes could simply drop the message). Such an
algorithm would induce $O(nm)=O(n^2)$ messages when
$m=O(n)$. Therefore, instead of seeking certainty, we investigate
algorithms that output the exact value of the aggregation function
\textit{with high probability (whp)}. In fact, our first contribution is to
prove a lower bound stating that at least $\Omega(n\log n)$ messages
are required to compute any multiparty function (not only aggregation)
in an accurate way and with high probability.
The lower bound proof leverages a strategy that consists for the adversary to control all the nodes to which a specific node sends its messages, whilst honest nodes hide as much
as possible the selection of nodes that will treat their votes.

Our second contribution is a near-optimal distributed
protocol for computing the output of any aggregation function
while protecting the privacy of individual inputs, and tolerating up to
$(\frac{1}{2} - \epsilon)n$ malicious nodes, for $n$ the number of nodes in
the network and $\epsilon>0$ any constant independent of $n$. 
 This protocol outputs the exact value of the aggregation function with high
probability as $n$ tends to infinity, has a global communication
cost of $O(n\log^3 n)$, and is balanced. It is based on a distributed version of the protocol presented in \cite{Awerbuch2009} to build an overlay of clusters of size $O(\log n)$, such that each of them contains a majority of
honest nodes (the construction of such a layout is a necessary condition to ensure the accuracy of our protocol.). It uses this
layout to compute the aggregation function using existing
cryptographic tools at the level of the clusters, that would otherwise
be too expensive (in terms of communication complexity) to run at the
level of the network.

We support our theoretical evaluations with an implementation of a distributed polling protocol, based on our aggregation protocol, on the Emulab testbed \cite{White+:osdi02}. We show its communication, computation, and time efficiency compared to a non-layout based secure aggregation protocol.

The rest of the paper is organized as follows.
In Section \ref{sec:model},
we discuss the model and related work. Afterwards, we prove a lower
bound on the minimal communication complexity of computing an
aggregation function in Section \ref{sec:impossibility}. The protocol
is described and analyzed in Section \ref{sec:protocol}.  Finally, we report on
the results of an experimental evaluation of the protocol in Section
\ref{exp_evaluation} and we conclude in Section~\ref{conclu}.


\section{Preliminaries and Background}
\label{sec:model}

\subsection{Model}
We consider a dynamic and distributed system composed of nodes that can join and leave at any time. We assume that the number of nodes in the system is always linear in a parameter $n$. Furthermore, some of the nodes (that we refer to hereafter as \emph{malicious}) might display an undesirable behaviour and cheat during the execution of the protocol in order to learn information of honest nodes, to perturb the functionality of the system, or simply to make it crash. We model these malicious nodes through a static and active adversary controlling and coordinating them (see \cite{Goldreich2001} for a formal definition). At any point in time, the adversary has a complete knowledge of the structure of the system while an honest node has only a local knowledge of its neighbourhood. When a node joins the network, this adversary can corrupt it to make it a malicious node, otherwise the node would be called honest. Of course, during the execution of a protocol, an honest node has no idea whether or not it is currently interacting with an honest node or a malicious node.
We assume that $\tau$, the fraction of malicious nodes controlled by the adversary, is less than $\frac{1}{2}-\epsilon$ for some constant $\frac{1}{2} > \epsilon>0$ independent of~$n$. In order to construct the overlay of the system, we rely on a distributed version of the protocol proposed in \cite{Awerbuch2009}. The analysis of this protocol has been conducted in the following setting: first, only the honest nodes are present in the network and there has been enough leave and join operations from these honest nodes, then the adversary has the malicious nodes join the network. 
After this initialisation phase, the nodes can leave and join arbitrarily, but still while conserving a maximum fraction $\tau$ malicious nodes in the system.

\begin{remark}
It has been proved that the protocol of \cite{Awerbuch2009} guarantees to maintain a partition of the nodes into clusters of size $\Omega(\log n)$, such that in each cluster, there is a majority of honest nodes under the hypothesis that the fraction of malicious nodes controlled by the adversary, is less than $\frac{1}{2}-\epsilon$ and that the adversary cannot force an honest node to leave the system. The partition can be made further resilient to other attacks (particularly denial-of-service attacks which result in honest nodes leaving the system), by using the protocol proposed in \cite{scheideler}. However, in this case, the fraction of malicious nodes controlled by the adversary, is required to be less than $\frac{1}{5}-\epsilon$. The results presented in this paper can be straightforwardly adapted to this latter situation.
\end{remark}

We further assume that each node has a
pair of secret/public keys for signature that is assigned to it by a (trusted) Certification Authority
(CA). Note that this CA is dedicated to key management and cannot be used to achieve other computations. Some decentralized implementations of such a CA are possible but they are out of the scope of this paper. More
precisely, when a new node joins the network, it receives a certificate signed by the CA that contains its public key and that can be shown to other nodes, as well as the corresponding private key. 
As the public key of a node is unique and chosen at random, it can be considered as being the identifier of this node. We also assume that all nodes in the system can communicate via pairwise secure channels (these secure channels can be obtained for instance via a public key
infrastructure or a key generation protocol), which means that all
the communications exchanged between two nodes are authenticated and confidential from the point of view of an external eavesdropper. 

Periodically, the nodes currently within the system compute, in a distributed manner,
an aggregation function $f$ that depends on individual inputs of
the nodes, $x_1, \ldots, x_n$, where $x_i$ is taken from a set of $k$
different possible values (i.e., $x_i \in
\{\nu_1,\dots,\nu_k\}$). This aggregation function can be used for
instance to compute a distributed polling or any global property of
the network that can be obtained by a linear combination of the local
inputs. In the rest of the paper, we assume that these values
correspond to integers. A simple example of such a function is the
computation of the average (or the sum) of these values (i.e., $f(x_1,
\ldots, x_n)=\frac{1}{n} \sum_{i=1}^n x_i$).

Due to the presence of malicious nodes, the computation needs to be achieved in a secure way, in the sense that it should offer some guarantees on the privacy (Definition~\ref{def:privacy-active-adversary}) of local inputs of honest nodes and on the correctness of the output, and this against any actions that the active adversary might do.

\begin{definition}[Privacy -- active adversary \cite{Goldreich2001}]
\label{def:privacy-active-adversary}
A distributed protocol is said to be \emph{private with respect to an active adversary} controlling a fraction $\tau$ of nodes, if this adversary cannot learn (except with negligible probability) more
information from the execution of the protocol than it could from its own input (i.e., the inputs of the nodes he controls) and the output of the protocol.
\end{definition}

Privacy can either be \emph{information-theoretic} (unconditional), if an adversary with unlimited computational power cannot break the privacy of the inputs of honest nodes, or \emph{cryptographic}, if it only holds against a computationally-bounded adversary that does not have enough computing resources to break a cryptographic assumption (such as factorising the product of two big prime numbers or solving the discrete logarithm problem). 
Furthermore, we aim at ensuring correctness (Definition~\ref{def:correctness-active-adversary}), even if the malicious nodes controlled by the adversary misbehave.

\begin{definition}[Correctness -- active adversary]
\label{def:correctness-active-adversary}
A distributed protocol is said to be \emph{correct with respect to an active adversary} controlling a fraction $\tau$ of the nodes, if the output of the protocol 
is guaranteed to be exact with high probability. 
\end{definition}


However,  malicious nodes cannot be denied from choosing their own inputs as they want as long as these inputs are valid (i.e. they are within the range of the possible ones).  Neither can we avoid the situation in which a malicious node simply refuses to participate to the protocol, in which case the protocol should be run without taking into account this node. If a malicious node leaves the protocol \emph{during} its execution then the protocol should be robust enough to carry out the computation instead of simply aborting. This corresponds to a type of denial-of-service attack on the global functionality computed by the protocol. 

Moreover, we seek \emph{scalable} protocols whose global communication cost and computational complexity is as close as possible to the lower bound $\Omega(n\log n)$ that we prove in Section \ref{sec:impossibility}. 
The term scalable reflects the property of a protocol to remain efficient when the size of the system increases, and hence its suitability to be implemented in large scale systems.
As the notion of scalability imposes a bound on the communication and computational complexity, it has an impact on the type of cryptographic techniques that can be used. For instance as mentioned in Section \ref{sec:relatedwork}, general results in secure multiparty computation can be used to compute any distributed function in a secure manner as long as the number of malicious participants is less than $n/3$ when aiming for information-theoretic security \cite{Chaum1988} or less $n/2$ for achieving cryptographic security \cite{Goldreich1987}. However, these protocols have a communication and computational complexities that are at least quadratic or cubic in the number of nodes $n$. 

Apart from scalability, we also aim at achieving a \emph{balanced} protocol (Definition~\ref{def:fair}) in which each node receives and sends approximately the same quantity of information.

\begin{definition}[$(C_{in},C_{out})$-balanced]
\label{def:fair}
A distributed protocol among $n$ nodes whose communication complexity is $C_{total}$ is said to be \emph{$(C_{in},C_{out})$-balanced}, if each node sends $O(C_{in}C_{total}/n)$ and receives $O(C_{out}C_{total}/n)$ bits of information.
\end{definition}

For instance, in a $(1,1)$-balanced protocol, each node sends and receives the same number of bits (up to a constant factor), whereas in a $(n,n)$-balanced protocol a single node may do all the work. 
In this paper, we will conisder as \emph{balanced} a protocol that is $(Poly(\log n),Poly(\log n))$-balanced.

\subsection{Homomorphic encryption}

In our work, we rely on a cryptographic primitive known as \emph{homomorphic encryption}, which allows to perform arithmetic operations (such as addition and/or multiplication) on encrypted values, thus protecting the privacy of the inputs of honest nodes.

\begin{definition}[Homomorphic cryptosystem]
\label{def:hom-enc}
Consider a public-key (asymmetric) cryptosystem where (1) $\mathsf{Enc}_{pk}(a)$ denotes the encryption of the message $a$ under the public key $pk$ and (2) $\mathsf{Dec}_{sk}(a)=a$ is the decryption of this message with the secret key $sk$. (In order to simplify the notation, we drop the indices and write $\mathsf{Enc}(a)$ instead of $\mathsf{Enc}_{pk}(a)$ and $\mathsf{Dec}(a)$ instead of $\mathsf{Dec}_{sk}(a)$ for the rest of the paper.) A cryptosystem is \emph{additively} homomorphic if there is an efficient operation $\oplus$ on two encrypted messages such that $\mathsf{Dec}(\mathsf{Enc}(a)\oplus\mathsf{Enc}(b)) = a + b$. Moreover, such an encryption scheme is called \emph{affine} if there is also an efficient scaling operation $\odot$ taking as input a cipher-text and a plain-text, such that $\mathsf{Dec}(\mathsf{Enc}(c) \odot a)=c \times a$.
\end{definition}

Paillier's cryptosystem \cite{Paillier1999} is an instance of a homomorphic encryption scheme that is both additive and affine. Moreover, Paillier's cryptosystem is also \emph{semantically secure} \cite{Goldreich2001}, which means that a computationally-bounded adversary cannot derive non-trivial information about the plain text $m$ encrypted from the cipher text $\mathsf{Enc}(m)$ and the public key $pk$.

\begin{definition}[Semantic security \cite{Goldreich2001}]
 An encryption scheme is said to be \emph{semantically secure} if a computationally-bounded adversary cannot derive non-trivial information about the plain text $m$ encrypted from the cipher text $\mathsf{Enc}(m)$ and the public key $pk$.
\end{definition}

For instance, a computationally-bounded adversary who is given two different cipher texts encrypted with the same key of a semantic cryptosystem cannot even decide with non-negligible probability if the two cipher texts correspond to the encryption of the same plain text or not. This is because a semantically secure cryptosystem is by essence \emph{probabilistic}, meaning that even if the same message is encrypted twice, the two resulting cipher-texts will be different with high probability. In practice, the semantic security is achieved by injecting in each possible operation of the cryptosystem (encryption, addition and multiplication) some fresh randomness. However, if two participants agree in advance to use the same seed for randomness, this renders the encryption scheme deterministic. For instance, if two different participants each add together two encrypted values, $\mathsf{Enc}(m_1)$ and $\mathsf{Enc}(m_2)$, using the addition operation of the cryptosystem, this normally results in two different cipher-texts (but still corresponding to the same plain-text), \emph{unless} the participants agree in advance to use the same randomness. 

In this paper, we use a \emph{threshold version} of the Paillier's cryptosystem due to D{\aa}mgard and Jurik
\cite{DamgardJ01}.

\begin{definition}[Threshold cryptosystem]
\label{def:threshold-cryptosystem}
A $(t,n)$ \emph{threshold cryptosystem} is a public cryptosystem where at least $t>1$ nodes out of $n$ need to actively cooperate in order to decrypt an encrypted message. In particular, no collusion of even $(t-1)$ nodes can decrypt a cipher text. However, any node may encrypt a value on its own using the public-key $pk$. After the threshold cryptosystem has been set up, each node $i$ gets as a result its own secret key $sk_i$ (for $1\leq i \leq n$).
\end{definition}

\subsection{Related Work -- Secure Multiparty Computation}
\label{sec:relatedwork}

The main goal of secure multiparty computation is to allow participants to compute in a distributed manner a joint function over their inputs while at the same time protecting the privacy of these inputs. This problem was first introduced in the bipartite setting in \cite{Yao1982} and as since became one of the most active fields of cryptography. Since the seminal paper of Yao, generic constructions have been developed for the multiparty setting that can be used to compute securely (in the cryptographic sense) any distributed function as long as the number of malicious nodes is strictly less than $n/2$ \cite{Goldreich1987}. For some functions, it has been proven that these protocols are optimal with respect to the number of malicious nodes that can be tolerated. Subsequently, complementary work \cite{CK91} has shown that the only functions that can be computed even when a majority of nodes are malicious (up to $n-1$) are those that consist in a XOR-combination of $n$ Boolean functions. In the information-theoretic setting, generic results for secure multiparty computation also exist that can tolerate up to $n/3$ malicious nodes without relying on any cryptographic assumptions \cite{Ben-Or1988,Chaum1988}. 
All these generic constructions have constant overheads (independent of the function to be computed), in terms of communication and computation, that are at least quadratic or cubic in the number of nodes \cite{BeerliovaTrubiniova2008}. 
However, for solving specific tasks it is often possible to develop more efficient protocols.
Furthermore, contrary to us, some of these constructions assume the availability of a secure broadcast channel, i.e it is assumed that a broadcast can be performed for a communication cost which is the size of the message. 
Recently, \cite{Hirt2010} has proposed a protocol for simulating such a channel that can tolerate up to $n/2$ malicious nodes with a global communication cost of $\Theta(n^2)$. We use this protocol as a subroutine in our approach to broadcast messages to a small set of nodes. Notice that this broadcast protocol is proved to be secure against an adaptive adversary, which is a stronger adversary than the one we consider (the static one). Therefore, one could argue that using this protocol is too strong and that we should rely on a simpler one.  


Assuming a particular structure of the underlying network can sometimes lead to efficiency gains. For instance, a recent protocol \cite{Ben-David2008} achieves the secure multiparty computation of any function in at most $8$ rounds of communication. This protocol uses a cluster composed of $m$ nodes (for $m<n$) to simulate a trusted central entity, has a global communication cost of $O(m^2 + nm)$ and can tolerate up to $m/2$ honest-but-curious nodes (a weaker form of adversary model in which nodes do not cheat actively but can share their data in order to infer new information). This protocol can be used to compute privately an aggregation function as long as the number of honest-but-curious nodes is small (i.e. $O(\log n)$) but does not aim at providing security against malicious nodes as we do. 
A related technique, frequently used in the context of scalable distributed consensus, is \emph{universe reduction}, 
which consists in selecting a small subset of the $n$ nodes, among which the fraction of malicious nodes is similar to the fraction of malicious nodes in the entire system with high probability. This subset of nodes can run expensive protocols and disseminate the results to the rest of the network. Up to our knowledge, no solution for this problem with a sub-quadratic global communication complexity has been found in presence of an adversary controlling a constant fraction of the nodes, without any simplifying assumptions. Feige's lightest bin protocol \cite{Feige1999} solves the problem when all node have access to a broadcast channel, and it has been adapted to the message passing model by \cite{King2006}, with a complexity polynomial in $n$, the original network size. Another solution proposed in \cite{Ben-Or2006} has a quasi-polynomial communication complexity (i.e. the total number of bits sent by honest nodes is $O(n^{\log n})$). Other protocols by \cite{Kapron2008} and \cite{King2010} claim sub-quadratic communication complexity against a malicious adversary controling less than $n/(6+\epsilon)$ (for some constant $\epsilon > 0$) and $n/3$ nodes respectively, 
but they assume that all nodes know all the identities of all the other nodes in the system, which in itself hides an $\Omega (n^2)$ communication complexity to propagate this information. 

Aggregation functions also have a clear connection with election problems. In particular for these problems, it is of paramount importance to protect the privacy of voters. In \cite{BFP+01}, a protocol was proposed that computes the outcome of an electronic election while providing cryptographic security for a global communication cost of $O(n)$. However contrary to our approach, this protocol requires the availability of a trusted entity. Another protocol \cite{Juels2010} achieves the property of coercion-resistance, i.e. a voter cannot produce any proof of the value of its vote, thus preventing any possibility to force or buy its vote. This protocol has a complexity of $O(n^2)$, where $n$ is the number of voters, which (quoting the authors) is not practical for large scale elections. Finally, our work is to be compared with \cite{Giurgiu2010}, in which the authors propose a protocol computing an $\sqrt{n}$-approximation of an aggregation function even in the presence $\sqrt{n}/\log n$ rational nodes (a slightly weaker form of adversary than malicious nodes) with a global communication cost of $O(n^{3/2})$ but without relying on cryptographic assumptions.

\section{Lower Bound}
\label{sec:impossibility}

In this section, we prove that no balanced algorithm can compute an aggregation function with a global communication complexity $o(n\log n)$ provided that the number of malicious nodes is linear in the number of nodes in the system.

\begin{theorem}\label{thm:low}
Given a protocol that computes a function in a distributed manner, whose inputs are held by $n$ nodes and among which  $\epsilon n$ are malicious for some positive constant $\epsilon<1$ (independent of $n$). If a fraction $cn$ of the nodes send no more than $\omega^+(n)$ messages (for some $\omega^+(n)=o(\log n)$
and constant $1>c>0$) and that no node receives more than $\omega^-(n)$ messages with $\omega^+(n)e^{\omega^+(n)}\omega^-(n) = o(n)$, then with high probability (in $n$) there is a node whose messages are all intercepted by malicious nodes.
\end{theorem}

\begin{proof}
If the structure of the overlay is fixed using deterministic arguments before the adversary chooses which nodes he corrupts, he can always target a particular node by corrupting the $\omega^+(n)$ nodes with whom this honest node will communicate. Similarly, if the set of nodes receiving all the messages of a given node are not chosen uniformly at random, the adversary may possibly use this bias to increase the chances that all the node in the determined set are malicious.

The best strategy for a node to avoid to send all its messages to malicious nodes, is to choose the recipients (we call recipient for $x$, a node that receives a message from node $x$) of its messages uniformly at random.
The objective of our proof is to show that even under this strategy, there is at least a node that will send all its messages to malicious ones. For this, we consider several disjoint sets, each containing all the recipients of a particular node. 

To find such disjoint sets of nodes, we proceed in a greedy way. Given a node $x_1$ that sends less than $\omega(n)$ messages, we set $S_1$ to be the set of recipients for $x_1$. By assumption, this set intersects at most $\omega^-(n)\omega^+(n)$ sets of recipients of other nodes. We then proceed recursively by discarding all the nodes sending their messages to these sets. At the end of the process, we obtain $k=cn/(\omega^-(n)\omega^+(n))$ nodes that send all their messages to disjoints sets $S_1,...,S_{k}$.

We now want to prove that with high probability, one of these sets is composed exclusively of malicious nodes. The adversary can choose arbitrarily the nodes he want to corrupt to be malicious. However, we assume that once its choice is made, it is fixed and the adversary cannot change it later. 
Recall that the sets $S_i$, $1\leq i \leq k$ have been chosen at random. Given the set $S_i$, we denote by $S=\cup_{i=1}^k S_i$ and $m_S$ the number of malicious nodes in $S$, the two following processes generate the same probability distribution:
\begin{enumerate}
\item Choose $\epsilon n$ malicious nodes, and then take disjoint sets $S_1,...,S_{k}$ among the $n$ possible ones.
\item Draw $m_S$, the number of malicious nodes in $S$, according to the binomial distribution $Bi(\epsilon,|S|)$, and then choose at random $m_S$ malicious nodes from $S$.
\end{enumerate}

Note that \emph{we do NOT assume that the malicious nodes are chosen at random}. Rather we state that if two processes lead to the same probability distribution,
 proving a result for one of the two processes implies the same result for the other process. Let us recall that the malicious nodes are corrupted by the adversary at the time they join the network.

When choosing $m_S$ malicious nodes from $S$, we have a non-independent probability for a node to be malicious or not, whereas, if we had independence, the proof would be easy.
Therefore our objective is now to introduce independence. Notice that choosing $m_S$ nodes at random in $S$ gives the same distribution as first choosing $m_S'<m_S$ nodes at random, and then $m_S-m_S'$ other nodes at random. Moreover, choosing $m_S'$ nodes at random has the same probability distribution as the process of choosing malicious nodes such that each node is malicious with probability $m_S'/S$, knowing that exactly $m_S'$ nodes are chosen. 

To summarize, the following process leads to the same probability distribution:
\begin{enumerate}
\item Choose $m_S$ according to the binomial distribution $Bi(\epsilon,|S|)$.
\item Choose at random and independently for each node if it is malicious with probability $m_S/2|S|$. This gives $m'_S$ malicious nodes.
\item With high probability $m_S>m_S'$. Choose $m_S-m_S'$ other malicious nodes at random.
\end{enumerate}

Step 3 can be proved using standard Chernoff's bound arguments (more precisely by showing that $m_S'$ is close to $m_S/2$), so if we can prove that after Step 1, there is a set containing exclusively malicious nodes, the theorem is proved. We therefore proceed in this direction.

We denote by $Mal$ the set containing all the malicious nodes. A set $S_i$ has the following probability to contains exclusively malicious nodes: $Prob(|S_i \cap Mal|=|S_i|)\geq (m_S/2|S|)^{\omega^+(n)}$.
$Q$ is the event in which no set is composed exclusively of malicious nodes and has the following probability: $Prob(Q)\leq (1- (m_S/2|S|)^{\omega^+(n)})^{k}$. However, with high probability, we have:
\begin{align*}
\lim_{n\rightarrow \infty} &(1- (m_S/2|S|)^{\omega^+(n)})^{k}\\
&=\lim_{n\rightarrow \infty} e^{-(m_S/2|S|)^{\omega^+(n)}*k}\\
&=\lim_{n\rightarrow \infty} e^{-cn(m_S/2|S|)^{\omega^+(n)}/(\omega^-(n)\omega^+(n))}\\
&=0
\end{align*} 
because 1) $\omega^+(n)e^{\omega^+(n)}\omega^-(n)=o(n)$ and 2) with high probability, $m_S/2|S| > \epsilon/4$.

Therefore, for big enough $n$, we will have an honest node that will send all its messages exclusively to malicious nodes. In this case, the malicious nodes can discard all its messages thus preventing the computation of the output to be exact. 
\end{proof}

\begin{corollary}\label{cor}
A $(Poly(\log n),n)$-balanced protocol that computes with high probability the exact value of a function in a distributed manner, whose inputs are hold by $n$ nodes among which $\epsilon n$ are malicious for some positive constant $\epsilon<1$ (independent of $n$), induces a total of $\Omega(n\log n)$ messages.
\end{corollary}

\begin{proof}
By contradiction, suppose that the protocol induces $o(n\log n)$ messages.
We have some $\omega^+(n)=o(\log n)$ and some constant $1>c>0$ such that $cn$ nodes send less than $\omega^+(n)$ messages. As the protocol is  $(Poly(\log n),n)$-balanced, no node receives more than $\omega^-=Poly(\log n)*o(\log n)=Poly(\log n)$ messages. These $\omega^+$ and $\omega^-$ verify the conditions of Theorem \ref{thm:low}, thus implying the existence of an honest node whose messages have been sent only to malicious nodes. Therefore as this node is ``surrounded'' by malicious nodes, they can decide to discard all its messages, thus preventing its input to be taken into account in the computation, and as a consequence the outcome cannot be exact. 
\end{proof}

\section{Distributed Aggregation Protocol}
\label{sec:protocol}

\subsection{Overview}
\label{high_level_description}

We now provide a high-level view of the protocol that distributively computes an aggregation function (see also Figure \ref{fig:principle}). Our protocol uses a distributed version of the protocol of \cite{Awerbuch2009}, which we describe in \ref{sec:layout}. It provides us with a partition of the $n$ nodes into $g$ clusters, $C_1, \ldots, C_g$, of equal size, which are organized in a ring such that all nodes of a cluster $C_{i}$ are connected to all nodes of the previous cluster ($C_{i-1\ mod\ g}$) and the next cluster ($C_{i+1\ mod\ g}$). Assuming that a threshold homomorphic cryptosystem has been set up and that all nodes know the public key $pk$ of this threshold cryptosystem, the aggregation function can be computed by the following steps:
\begin{enumerate}
\item Each node encrypts its input under the public key of the threshold homomorphic cryptosystem and securely broadcasts it to all the other nodes in its cluster. Secure broadcast ensures that (i) all honest players eventually output an identical message, whatever the sender does (consistency), and (ii) this output is the sender's message, if it is honest (validity) \cite{Hirt2010}.
\item) The nodes in a given cluster agree on a common random string $rand$, and then each node adds the encrypted input it received from its own cluster, using the addition operation of the homomorphic cryptosystem by taking $rand$ as the randomness injected into the homomorphic encryption. The result of this addition is called the \emph{local aggregate} and is the same for each honest node of the cluster.
\item Starting from cluster $C_1$, the nodes add their local aggregate with the partial aggregate they received from the previous cluster, with the exception of the nodes of cluster $C_1$ that have received no partial aggregate and hence skip this phase. The nodes of the cluster then send the result of this aggregation (the new \emph{partial aggregate}) to all the nodes of the next cluster. The nodes of the next cluster then consider as the new partial aggregate the encrypted value appearing in a majority (to correct potentially inconsistent messages that have been sent by malicious nodes). This process is repeated $g-1$ times until the partial aggregate reaches the last cluster. Then, we say that the \emph{partial aggregate} has become the \emph{global aggregate}.
\item The nodes of the last cluster perform a threshold decryption of the \emph{global aggregate} revealing the output of the aggregation function.
\end{enumerate}


\begin{figure}[t]
\begin{center}
\subfigure{\scalebox{0.3}{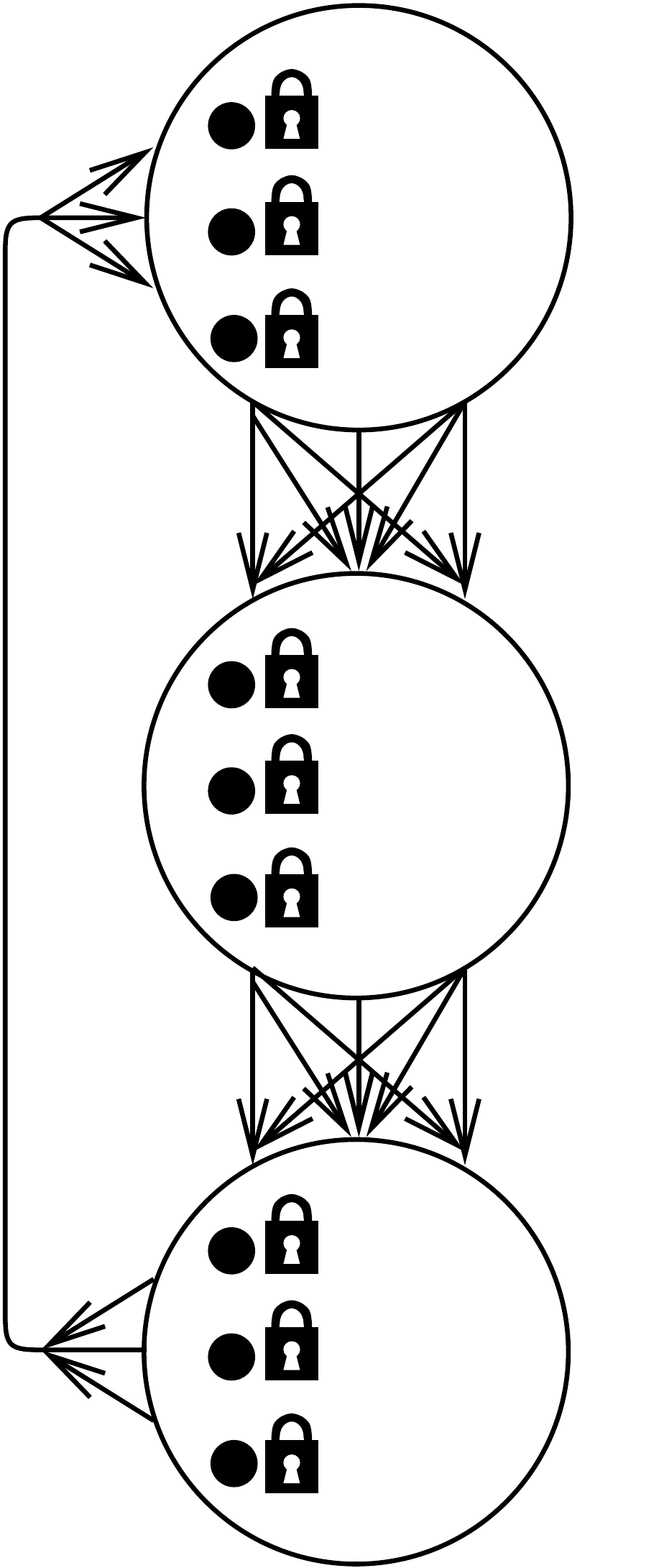}}  \hspace{2cm}
\subfigure{\scalebox{0.3}{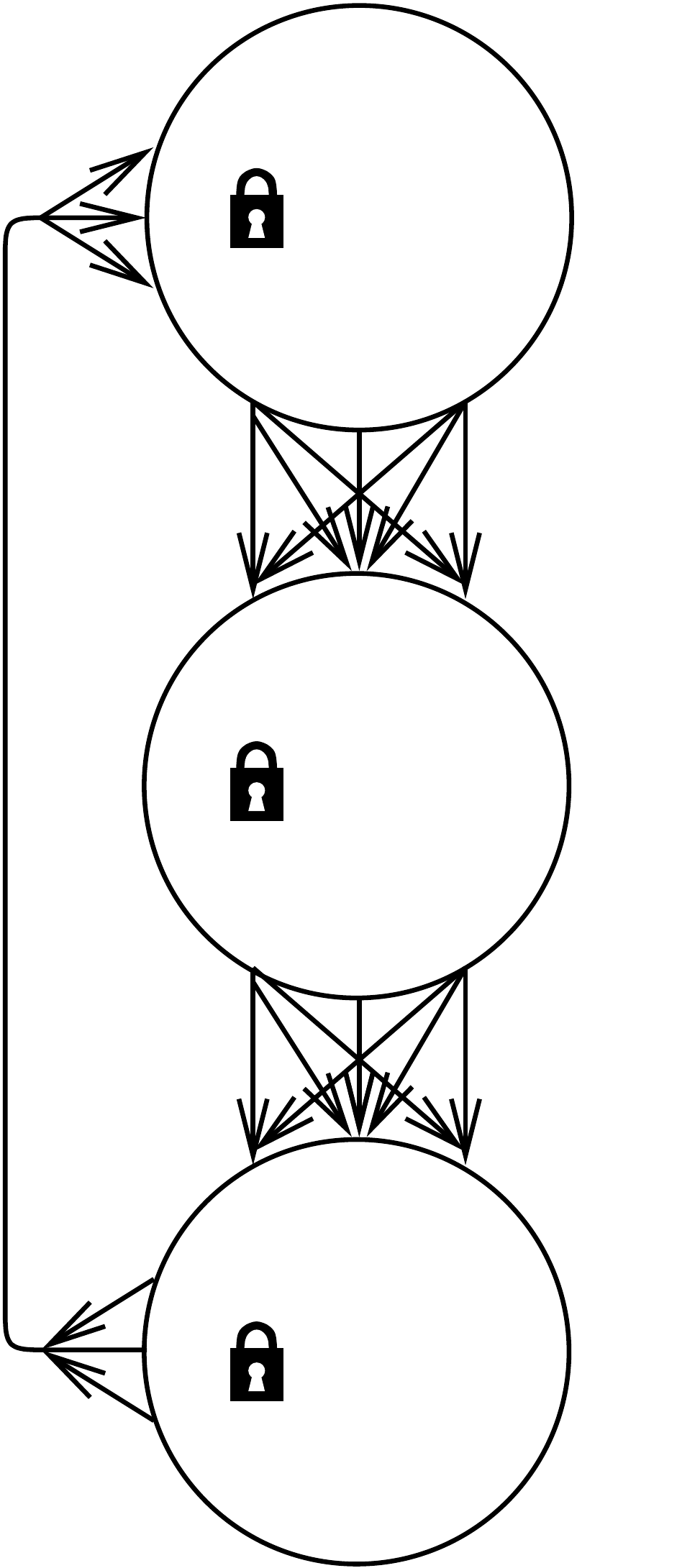}}  \hspace{2cm}
\subfigure{\scalebox{0.3}{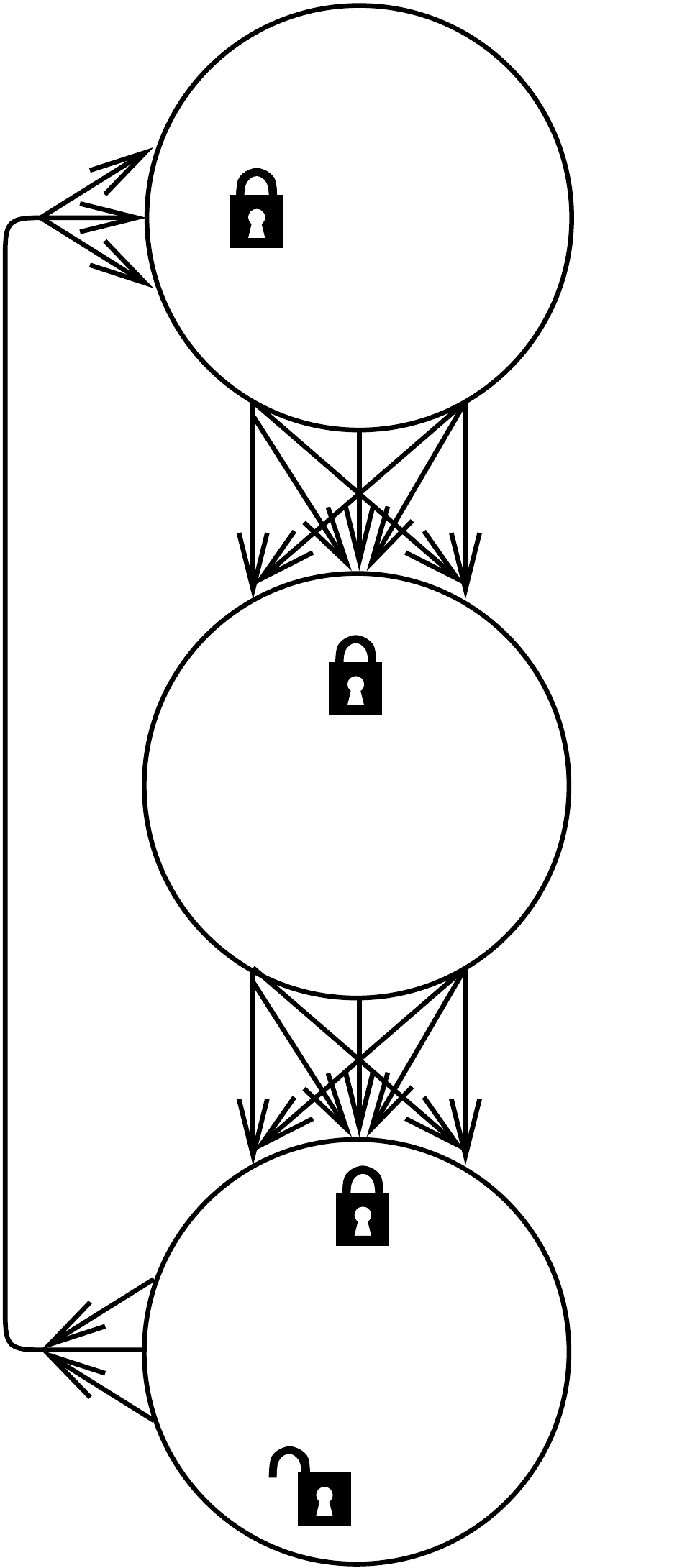}} 
\caption{Main idea of the algorithm. First, all nodes start by encrypting their inputs and broadcasting them to all the other nodes of the cluster. Within each cluster, each node computes the local aggregate ($\sum a_i$, $\sum b_i$ and $\sum c_i$), which is then propagated along the ring. After this, the nodes of cluster $B$ know $\sum a_i + \sum b_i$, and those in $C$, $\sum a_i + \sum b_i + \sum c_i$. The nodes of the last cluster (here $C$) collaborate to perform the threshold decryption and to obtain $S$, the output of the aggregation function.}
\label{fig:principle}
\end{center}\end{figure}

\begin{remark}[Bootstrapping]
The distributed aggregation protocol is meant to be used in a dynamic network, but still a certain stability of the partition of the cluster is required while this computation is taking place. To achieve this, when a node wants to start the computation of an aggregation function, it launches the protocol by broadcasting a time window during which the aggregation protocol will be performed. Any join operation occurring during this time window will be postponed, and similarly, any operation induced by the "leave rule" (cf Section \ref{sec:layout}) , except a node leaving itself, will be postponed. Using the layout, any node can perform a broadcast for a communication cost of $O(n\log n)$ by first broadcasting it to its cluster, and then the clusters recursively forward the message along the ring. 
\end{remark}

\subsection{The overlay}\label{sec:layout}
For the purpose of the aggregation protocol, we need the nodes to be organized in a layout consisting in a partition of the nodes into clusters $C_1, \ldots, C_g$ of size $O(\log n)$. To achieve this, we assume that the nodes join the network according to a distributed version of the protocol presented in \cite{Awerbuch2009}. This protocol uses a central authority and we explain how to discard this requirement at the end of the section, but first we describe its functionnality. The central athority is used to assign to each node a random number between 0 and 1, number which stands for the position of the node in the segment $[0,1)$. 
 By inducing churn when a node joins the system, the protocols ensures that each segment of size $c\log(n)/n$, for some specific $c$, contains a majority of honest nodes as long as the adversary controls at most a fraction of $\frac{1}{2}-\epsilon$ of malicious nodes, for some constant $\frac{1}{2} > \epsilon > 0$. The clusters $C_1, \ldots, C_g$ are composed of the nodes whose positions are in the respective segments $[0, c\log n/n), \dots, [1-c\log n/n,1)$. 

In order to efficiently distribute the protocol, we first arrange the \emph{clusters} in a Chord-like overlay \cite{Young2010}. The adapted join-leave protocol must further ensure that, for all $1\leq i \leq g$, the nodes from cluster $C_i$ know all the nodes from $C_{i+2^j}$, for all $1\leq j \leq \log_2 g$, resulting in $O(\log^2 n)$ connections per node (c.f. Figure \ref{Chord}).

\begin{figure}
\begin{center}

\scalebox{0.4}{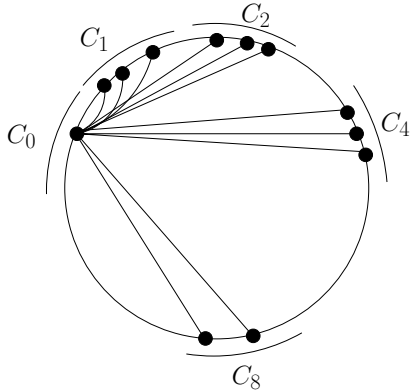}
\caption{Chord overlay.}\label{Chord}
\end{center}
\end{figure}
As in \cite{Awerbuch2009}, we assume that initially, the network is composed of only honest nodes (this assumption is also typical in sensor networks or in networks created by trusted peers).
This assumption allows to bootstrap the layout construction by building an initial overlay by assigning to each initial node a position in $[0,1)$. We now describe how to adapt the join rule of \cite{Awerbuch2009}, known as the cuckoo rule, and how to maintain the overlay. We further show how to adapt the leave rule of \cite{scheideler} that is required when the adversary can force a node to leave the network (for instance through a denial-of-service attack).

We will rely on two subroutines to distribute these rules.
\begin{enumerate}
\item {Inter-cluster communication:} A node from a cluster $C$ receiving a message from a node of a cluster $C'$ accepts it if and only if $C'$ is a neighbor of $C$ in the Chord overlay, and at least half of the nodes of $C'$ have sent the same message.
\item {Random Number Generation:} To generate a number at random (typically between 0 and $\log n$), the nodes of a cluster $C$ collaborate as follows: each node of $C$ chooses a random number in the desired range, and commits it to the other nodes of $C$ via a secure broadcast. 
Once all nodes have done so or after a given time-out if some nodes do not participate, the committed values are revealed by sending the decryption keys again using the secure broadcast. All the valid random numbers are added modulo some predecided upper bound ($\log n$ in our example) and the result corresponds to the generated random number, which is agreed upon by all the honest nodes.
\end{enumerate}

We can now explain how the cuckko rule is modified in order to discard the requirement of a central authority

\paragraph{Join rule or cuckoo rule:} when a node $x$ joins the network, we suppose that it can come in contact with one of the node of the network. 
The node gives the composition of an arbitrary cluster to $x$ (this cluster is chosen among the clusters that the contacted node knows). In turn, $x$ contacts the whole cluster (i.e. all its members) with a join request. 
Afterwards, this cluster starts performing the cuckoo rule from \cite{scheideler}, which corresponds to choosing a position at random in $[0,1)$ for $x$. 
 Then the nodes of the cluster containing the chosen position, which we call $C$, are informed via messages routed using the Chord overlay that $x$ is inserted at this position. 
At this point, extra churn is induced, and for a constant $k>\frac{1}{\epsilon}$, $C$ chooses a new random position for all the nodes of $C$ whose positions are in a segment of length $k/n$ containing the previously chosen position (c.f. \cite{scheideler} for more details). Whenever a node $x$ of a cluster $C'$ is required to change its position to join a cluster $C"$, all the nodes that were adjacent or that become adjacent to this node, are informed of the change by messages sent by the nodes of $C'$ and $C"$ respectively. 
This is important since we need that at any time, if two clusters  $C$ and  $C'$ are adjacent in the Chord overlay, all nodes of $C$ should know the exact composition of $C'$ to decide whether or not it accepts a message from nodes of $C'$ during inter-cluster communication.

\ \\

When the adversary can force any honest node to leave the network, it is shown in \cite{scheideler} that churn needs to be introduced whenever a node leaves. For this, they designed a leave rule that can be adapted as follows:
\paragraph{Leave rule:}
when a node at position $p$ leaves the system, the cluster $C$ to which it belongs chooses a segment of $[0,1)$ of length $k/n$ (recall that $k>\frac{1}{\epsilon}$). All the nodes of this segment replace the nodes of a randomly chosen segment of same length contained in $C$. The old nodes of the replaced segment are moved to position chosen at random in $[0,1)$. Similar to the join operation, all the nodes that were adjacent or that become adjacent to the moved nodes are informed of these changes.


With high probability, the communication overhead of the distributed version of this protocol is $O(\log^3 n)$ per join or leave operation. Indeed, with high probability, the number of nodes in a segment of size $k/n$ is $O(k)$, therefore with high probability the protocol needs to generate a constant number of random numbers (for a constant depending on $k$ and therefore on $\epsilon$).
This protocol ensures that each cluster contains $\Omega(\log n)$ nodes, among which there is a majority of honest nodes as long as $\tau \leq 1/2-\epsilon$ when one do not consider denial-of-service \cite{Awerbuch2009}, or for $\tau \leq 1/5-\epsilon$ otherwise \cite{scheideler}.

\subsection{Aggregation computation}\label{sec:agregation}

In this section, we present the second phase of the protocol that computes the aggregation function. This protocol is optimal up to a logarithmic factor in terms of scalability (i.e. communication and computational complexity) and achieves perfect privacy, correctness against a computationally-bounded adversary controlling at most $(1/2-\epsilon)n$ malicious nodes, for a constant $\frac{1}{2} > \epsilon > 0$ independent of $n$ and is balanced. 
For this, we rely on cryptographic techniques that would be too expensive to use on the whole network, but whose costs are efficient when computed at the granularity of a cluster in the structured overlay.

\subsubsection{Setting up the threshold cryptosystem (Step 1)} One particular cluster (that we refer thereafter as the \emph{threshold cluster}) will be in charge of setting up the threshold cryptosystem. For instance, we can decide arbitrarily this threshold cluster to be the last one on the ring. This setup phase requires that all the nodes of this threshold cluster engage themselves in a distributed key generation protocol \cite{Nishide2011} for the threshold homomorphic cryptosystem \cite{DamgardJ01}. 
At the end of this key generation phase, all the nodes in the threshold cluster receive the same public key $pk$ and each one of them gets as private input a different secret key, respectively $sk_1,\dots,sk_{k\log n}$. The threshold cryptosystem is such that any node can encrypt a value using the public key $pk$ but that the decryption of a homomorphically encrypted value requires the active cooperation of at least $t$ of the nodes. In our case, we can set $t$ to be close to $\frac{k\log n}{2}$ to ensure that there will be enough honest nodes to cooperate for the final decryption of the result at the end of the protocol. Once the threshold cryptosystem has been set up, the public key $pk$ is communicated in the network cluster by cluster, following the ring structure. Thereafter, when we say that \emph{a cluster communicates} something \emph{to the next cluster}, we mean that all the nodes in the current cluster communicate the same value to all the nodes in the next cluster, which results in a communication cost of $O(\log^2 n)$. Once, a node in one cluster has received all the $k\log n$ messages from the previous cluster, it decides to keep as final value for this particular round of intercluster communication the value obtained by performing a majority vote. This value always corresponds to the unique value sent by all the honest nodes of the previous cluster, which are themselves in majority inside this cluster due to the property of the structured overlay.

\subsubsection{Local aggregation (Step 2)} The nodes within a cluster communicate their value encrypted under the public key $pk$ of the homomorphic encryption system to all the other nodes of the cluster through a secure broadcast channel. When a node sends a message in such a channel, it is received by all the other nodes of the cluster, who know the identity of the sender of this message. The complexity of constructing such a channel is polynomial in the number of  nodes of the cluster, which in our case is $O(\log^2 n)$ and can be obtained for instance by using a recent construction of Hirt and Zikas \cite{Hirt2010} as long as the number of malicious nodes is less than half of the nodes, which is by assumption the case in our protocol. All the members of the current cluster encrypt their inputs under the public key $pk$ and broadcast this encrypted value to all the members of the cluster along with a non-interactive zero-knowledge proof that this input is valid \cite{Yuen2009} 
(i.e. chosen from the range of valid values $\{\nu_1,\dots,\nu_k\}$). The purpose of this non-interactive zero-knowledge proof is to prevent an adversary from tampering with the output of the protocol by providing as input an arbitrary value that is outside the range of the possible ones. The privacy of the inputs is protected by the semantic property of the cryptosystem. Once all the nodes in the cluster have received the $\log n$ encrypted inputs from the other members, they add them together using the additive property of the homomorphic cryptosystem. This  produces as output the encryption of the sum of all the inputs within this cluster. With respect to the randomness used in the addition operation of the homomorphic encryption, we assume that all the nodes have agreed on a common value $rand$. This value can be obtained for instance by concatenating all the encrypted inputs and then hashing them to obtain a unique random value.

\subsubsection{Global aggregation and threshold decryption (Step 3).}
 The global protocol  proceeds iteratively during $\frac{n}{\log n}$ iterations, cluster by cluster, starting from the first cluster and following the ring structure of the overlay. At each iteration, the current cluster performs the local aggregation as described above and also aggregates the sum of the local inputs with the encrypted value received from the previous cluster to compute the partial aggregate. This partial aggregate corresponds to the global aggregation of the inputs of the clusters visited so far. As mentioned previously, the encrypted value received from the previous cluster can be computed locally by each node of the current cluster by making a majority vote on the $k\log n$ messages received from the previous cluster. Once the threshold cluster has been reached (i.e. the last cluster on the ring), the members of this cluster add their local aggregated values to the encrypted value received from the previous cluster. This produces an encryption of the sum of all the values. Finally, the members of the threshold cluster cooperate to decrypt this encrypted value by using their private keys. Along with their decryption shares, the nodes send a non-interactive zero-knowledge proof showing that they have computed a valid decryption share of the final outcome \cite{DamgardJ01}. As the number of nodes needed to decrypt successfully is $t=\frac{k\log n}{2}$ and that we have a majority of honest nodes in the cluster, this threshold decryption is guaranteed to be successful. The final output is forwarded cluster by cluster, following the ring structure of the overlay. 

\subsection{Analysis of the protocol}
\label{analysis}
During the protocol, due to the use of the Paillier's cryptosystem all the messages exchanges will be of constant size. The value of the constant is directly proportional to a security parameter of the cryptosystem and will not be larger than $\log n$. In the following, we count it as $O(\log n)$.

\begin{lemma}[Communication cost of the protocol]
The protocol for the distributed computation of an aggregation function has an overall communication complexity of $O(n \log^3 n)$ and is $(Poly(\log n),Poly(\log n))$-balanced in the sense that no node sends or receives more than $Poly(\log n)$ bits of information for an average of $O(\log^3 n)$ bits of information per node.
\end{lemma}

\begin{proof}
During the setup of the threshold cryptosystem (Step 1) and the threshold decryption (Step 3), only one cluster is involved and the communication cost of the primitives used (threshold cryptosystem setup, secure broadcast, and threshold decryption) is polynomial in the size of the cluster, which corresponds to a communication complexity of $O(Poly(\log n))$. 
During the local aggregation (Step 2), in each cluster, each node broadcasts its encrypted input and a broadcast induces a communication cost of $O(\log^3 n)$. As there are $O(n/\log n)$ clusters with $O(\log n)$ nodes in each cluster, it results in a communication cost of $O(n\log^3 n)$. Finally, the intercluster communication requires that all the $n$ nodes send $O(\log n)$ messages, each of size $O(\log n)$, to the nodes of the next cluster, resulting in a communication cost of $O(n\log^2 n)$. As a result, the protocol is dominated by Step 2 (the local aggregation), which leads to a global communication cost of $O(n \log^3 n)$. Moreover, it is easy to see from the description of the protocol that it is balanced in the sense that it requires $O(Poly(\log n))$ communications from each node. 


\end{proof}

\begin{lemma}[Security of the protocol]
The protocol for the distributed computation of an aggregation function ensures perfect security against a computationally-bounded adversary, tolerates $(\frac 12-\epsilon)n$ malicious nodes for any constant $frac{1}{2} > \epsilon > 0$ (not depending on $n$), and outputs the exact value of the aggregated function with high probability.
\end{lemma}
\begin{proof}
The privacy of the inputs of individual nodes is protected by the use of a cryptosystem that is semantically secure and also by the fact that the adversary cannot decrypt the partial aggregate because it does not know the necessary $t$ secret keys of the threshold cryptosystem to do so. The correctness is ensured by a combination of several techniques. First, the non-interactive zero-knowledge proof that each node issues along with the encrypted version of its value guarantees that the malicious nodes cannot cheat by choosing their values outside the range of the possible ones. Second, the secure broadcast ensures that honest nodes in each cluster have the same local aggregate. Third, the majority vote performed every time the nodes of a cluster communicate with the nodes of the next cluster along with the fact that there is a majority of honest nodes in each cluster due to the construction of the structured overlay ensures that the correctness of the partial aggregate will be preserved during the whole computation. Finally, the non-interactive zero-knowledge proof of the validity of the partial shares during the threshold decryption prevent the malicious nodes from altering the output during the last step of the protocol. 
\end{proof}

\section{Experimental Evaluation}

In this section, we evaluate the practical performance of the protocol through experimentations. More precisely, we compare the communication and time complexity of our protocol to a non-layout based protocol. We implemented a distributed polling protocol with two choices to select from, which is an instance of an aggregation protocol, where the nodes' votes are encoded into inputs for the aggregation protocol and decoded at the end. 
The experiments were run on Emulab \cite{White+:osdi02}, a distributed testbed allowing the user to choose a specific network topology using NS2 configuration file. In each experiment, we used up to 80 \emph{pc3000} machines, Dell PowerEdge 2850s systems, with a single 3 GHz Xeon processor, 2 GB of RAM, and 4 available network interfaces. Each machine ran Fedora 8 and hosted 10 nodes at a time. The nodes were connected to the router in a star topology, setting the maximum network bandwidth to $1000$Mb, and the communication relied on UDP. We used the ``Authenticated Double-Echo Broadcast algorithm'' (as described in \cite{introreliable}) for secure broadcast and the ``Paillier Threshold Encryption Toolbox'' \cite{paillieronline} for threshold encryption. The nodes adjust their message sending rate to be uniformly distributed between 0 and 2 seconds.\\
The only protocol that comes close to our protocol in terms of guarantees is a non-layout based protocol in which each node 1) securely broadcasts its encrypted input to all the other nodes, 2) combines the values it receives, 3) securely broadcasts its decryption share to the others, and 4) combines the decryption shares to obtain the output. Such a protocol provides privacy and correctness with certainty against an honest majority for a communication complexity of this protocol is $O(n^3)$. Therefore, even for medium sized networks under test (200-800 nodes), it reaches prohibitive complexities, in the order of hundreds of millions of messages exchanged. Accordingly, to evaluate this protocol in comparison to ours, we measured the complexity of single secure broadcast instances. Since the broadcast instances are assumed to be run in sequence to avoid congestion, we can simply add the communication and time complexities. For ease of implementation, we used a centralized key generation authority for setting up the threshold Paillier cryptosystem, whose modulus is fixed to 1024 bits. When evaluating our protocol, we included the overhead of generating keys for the threshold cryptosystem, whereas we assume that the threshold cryptosystem is already set up for non-layout based protocol. This is justified by the fact that these keys can be kept longer than in our case, where the keys need to be generated each time the layout is changed. Moreover, we take the cluster size to be $20*\log n$, which was found to be adequate for an adversary fraction of $\frac{3}{10}$.


\emph{ a) Communication Cost:}

Communication cost is the total number of Megabytes sent during the protocol execution.  Figure~\ref{fig:experiments}.a depicts the global communication cost for the two protocols, on a semi-log scale, with a varying network size.
It is apparent that our protocol (labelled as $DA$) features a much smaller cost growth than the non-layout based one (labelled as $NL$) as the number of nodes increases. With only 200 nodes in the network, the $NL$ protocol is expected to require around 30 GB overall! To get further insight, we plot in Figure~\ref{fig:experiments}.b the communication cost per node. It is noticeable that the growth now is negligible with $DA$ protocol, with  a 20 MB cost per node while the trend remains as before with $NL$, reaching 600 MB per node cost in a 400 nodes network.

\emph{ b) Execution Duration:}
In this experiment, for each network size, we average the time taken by a complete execution of the two protocols over all the nodes. This is represented by the bars in Figure~\ref{fig:experiments}.c. Similar to the previous experiment, the growth is much smoother with our $DA$ protocol. While the $NL$ protocol reaches tens of hours of execution duration, our protocol remains within the scale of minutes. The major part of this duration is communication delays as apparent from Figure~\ref{fig:experiments}.d, which shows the duration of the major computationally expensive steps. Particularly, we plot the durations of the vote encryption, share computation, and vote decryption (share combination) for the different decryption network sizes. The first four correspond to the decryption cluster size in our protocol consecutively mapped to the same network sizes as before (200, 400, 600, 800). This figure highlights  the efficiency of our protocol and shows that computational complexity is small when compared to the global time complexity. It further indicates that decryption is the most expensive step in such protocols, thus supporting our technique in delegating this task to small clusters.\\
It is important to note that the goal of this implementation was to provide the comparison between the two protocols and not to provide accurate performance measures of our protocol under various scenarios. Otherwise, several optimizations on the message size and message complexities can be integrated to reduce the complexity in real-world implementations.


\label{exp_evaluation}

\begin{figure}
\begin{center}

\subfigure[Global communication cost]{\scalebox{0.35}{\includegraphics{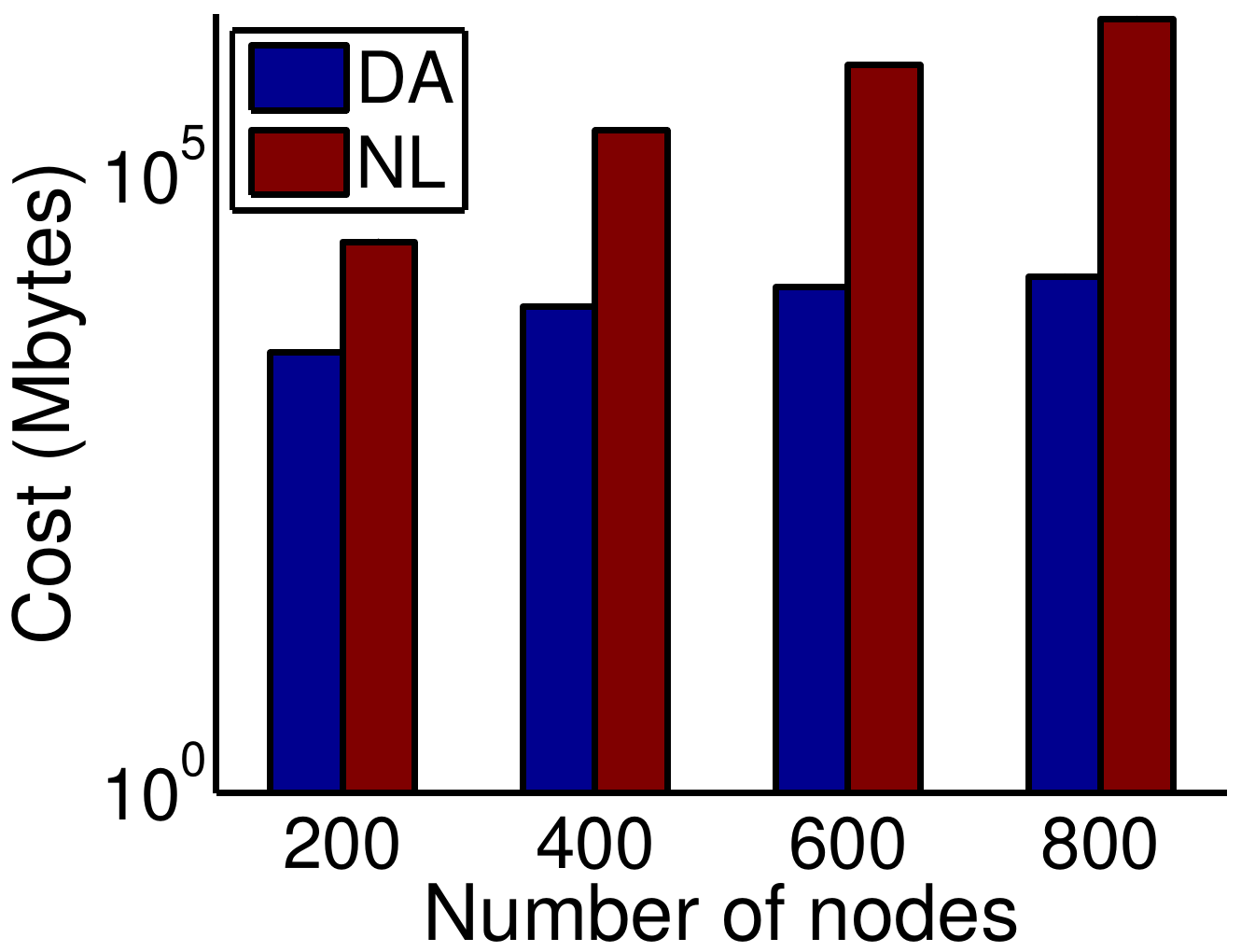}}} \hspace{2cm}
\subfigure[Communication cost per node]{\scalebox{0.45}{\includegraphics{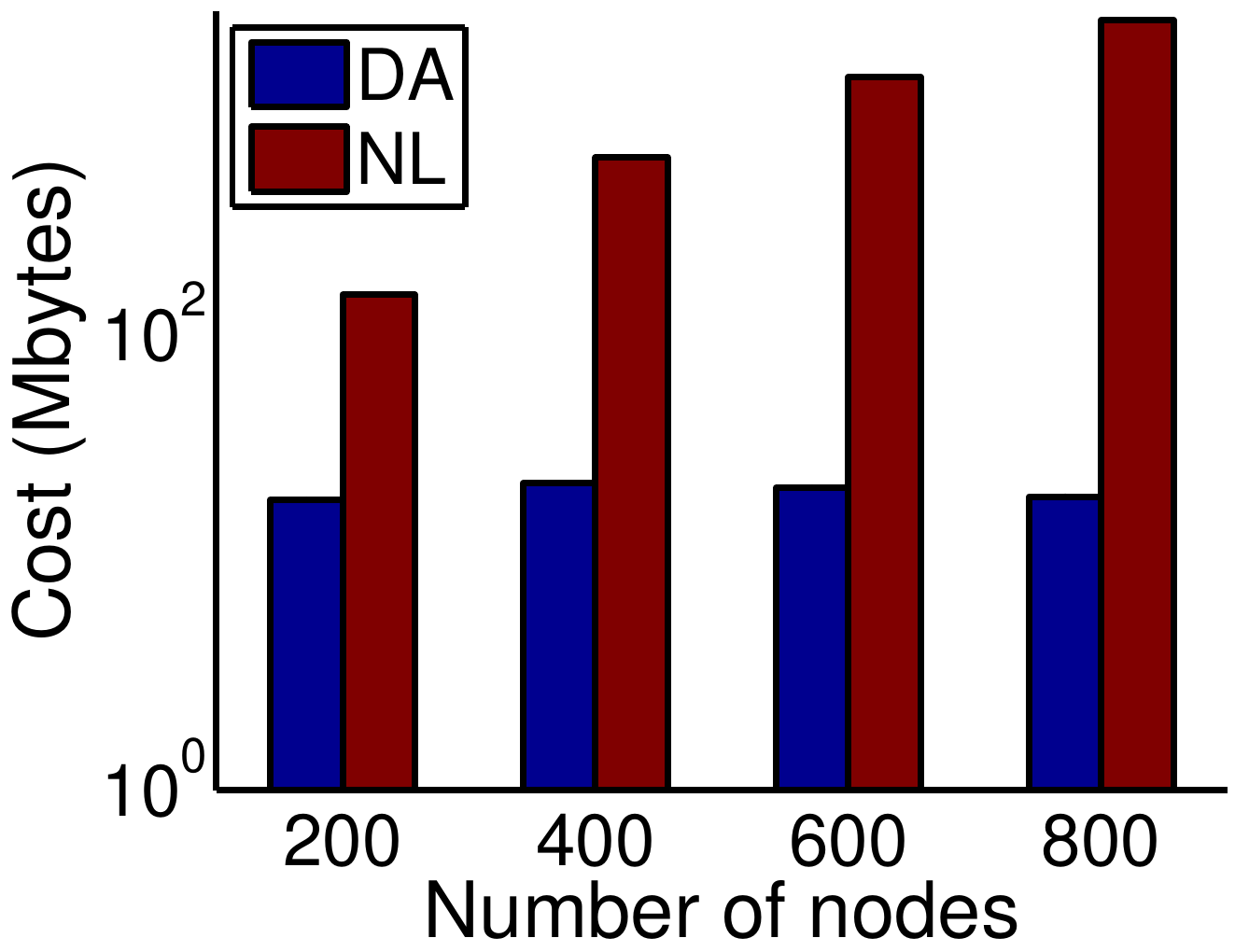}}}\\
\subfigure[Time complexity]{\scalebox{0.45}{\includegraphics{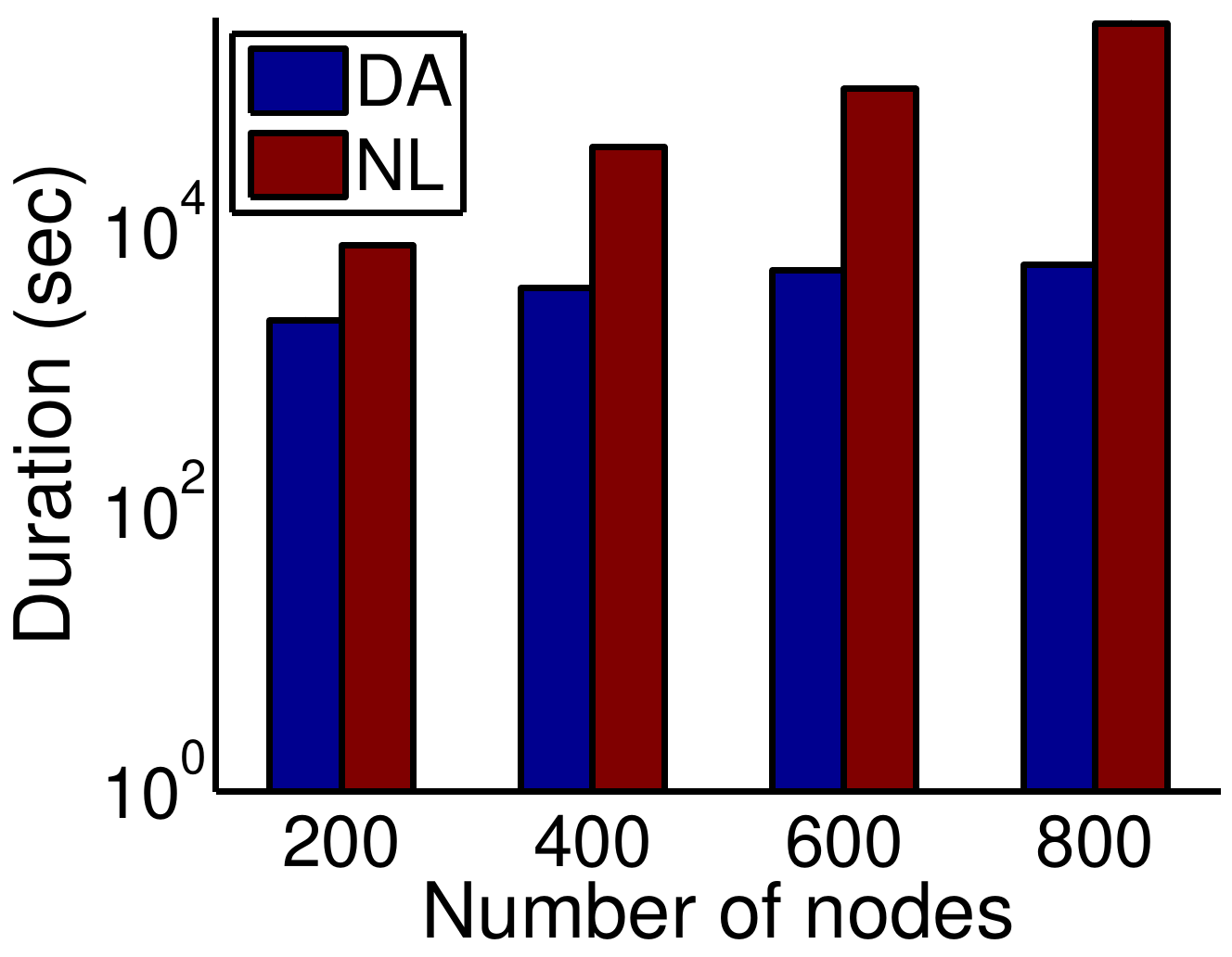}}} \hspace{2cm}
\subfigure[Computational complexity]{\scalebox{0.45}{\includegraphics{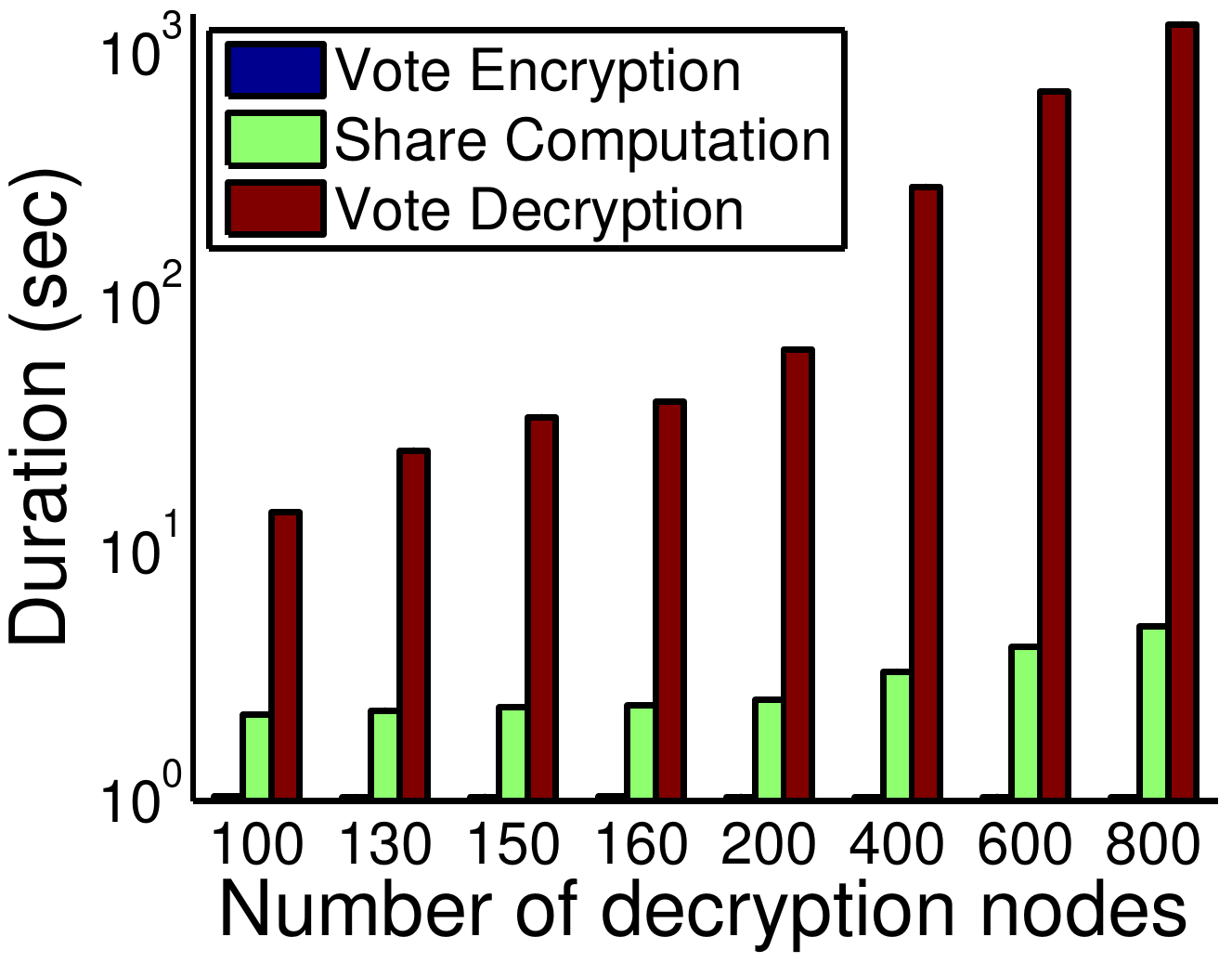}}}
\caption{Experimental evaluation of our protocol (DA) against a non-layout based one (NL).}\label{fig:experiments}
\end{center}
\end{figure}

\section{Conclusion}
\label{conclu}

In this paper, we have proposed a distributed protocol that computes aggregation functions in a distributed, scalable and secure way. The complexity of our protocol is drastically lower than those of previously known algorithms and within a factor $\log^2 n$ of the optimal. Note that using a binary tree structure instead of a ring allows to reduce the number of communication rounds from $\frac{n}{\log n}$ currently to $\log^2 n$. This can be obtained by a slight modification of the protocol for constructing the structured overlay but would not change the overall communication of the protocol. We leave as open the possibility of deriving a protocol closing the gap between the lower bound of $\Omega(n\log n)$ and the upper bound $O(n\log^3 n)$ achieved by our algorithm. We also leave as future work the possibility of adapting our protocol to achieve other multiparty computations than aggregation functions and the possibility of tolerating strictly less than $< n/2$ instead of $(1/2-\epsilon)n$.

\bibliographystyle{alpha}
\bibliography{Cryptography.bib,Hamza.bib}

\newcommand{\etalchar}[1]{$^{#1}$}
\begin{thebibliography}{YHM{\etalchar{+}}09}

\bibitem[AS07]{scheideler}
Baruch Awerbuch and Christian Scheideler.
\newblock {Towards scalable and robust overlay networks}.
\newblock In {\em Proc. 6th Int. Workshop on Peer-To-Peer Systems (IPTPS)}.
  Citeseer, 2007.

\bibitem[AS09]{Awerbuch2009}
Baruch Awerbuch and Christian Scheideler.
\newblock {Towards a scalable and robust DHT}.
\newblock {\em Theory of Computing Systems}, 45(2):234--260, February 2009.

\bibitem[BDNP08]{Ben-David2008}
A.~Ben-David, N.~Nisan, and Benny Pinkas.
\newblock {FairplayMP: a system for secure multi-party computation}.
\newblock In {\em Proceedings of the 15th ACM conference on Computer and
  communications security}, pages 257--266. ACM, 2008.

\bibitem[BFP{\etalchar{+}}01]{BFP+01}
O.~Baudron, P.A. Fouque, D.~Pointcheval, J.~Stern, and G.~Poupard.
\newblock {Practical multi-candidate election system}.
\newblock In {\em Proceedings of the twentieth annual ACM symposium on
  Principles of distributed computing}, pages 274--283. ACM, 2001.

\bibitem[BOG88]{Ben-Or1988}
M~Ben-Or and S~Goldwasser.
\newblock {Completeness theorems for non-cryptographic fault-tolerant
  distributed computation}.
\newblock {\em on Theory of computing}, 1988.

\bibitem[BOPV06]{Ben-Or2006}
Michael Ben-Or, Elan Pavlov, and Vinod Vaikuntanathan.
\newblock Byzantine agreement in the full-information model in o(log n) rounds.
\newblock In {\em Proceedings of the thirty-eighth annual ACM symposium on
  Theory of computing}, STOC '06, pages 179--186, New York, NY, USA, 2006. ACM.

\bibitem[BTH08]{BeerliovaTrubiniova2008}
Z.~Beerliov\'{a}-Trubiniov\'{a} and M.~Hirt.
\newblock {Perfectly-secure MPC with linear communication complexity}.
\newblock In {\em Proceedings of the 5th conference on Theory of cryptography},
  pages 213--230. Springer-Verlag, 2008.

\bibitem[CCD88]{Chaum1988}
David Chaum, Claude Cr\'{e}peau, and Ivan Damg{\aa}rd.
\newblock {Multiparty unconditionally secure protocols}.
\newblock {\em Proceedings of the twentieth annual ACM symposium on Theory of
  computing - STOC '88}, pages 11--19, 1988.

\bibitem[CGR11]{introreliable}
Christian Cachin, Rachid Guerraoui, and Lu\'{\i}s Rodrigues.
\newblock {\em Introduction to Reliable and Secure Distributed Programming (2.
  ed.)}.
\newblock Springer, 2011.

\bibitem[CK91]{CK91}
Benny Chor and Eyal Kushilevitz.
\newblock {A Zero-One Law for Boolean Privacy}.
\newblock {\em SIAM Journal on Discrete Mathematics}, 4(1):36, 1991.

\bibitem[DJ01]{DamgardJ01}
I~Damg{\aa}rd and M~Jurik.
\newblock {A Generalisation, a Simplification and Some Applications of
  Paillier's Probabilistic Public-Key System}.
\newblock In {\em Public Key Cryptography}, number December, pages 119--136.
  Springer, 2001.

\bibitem[DL11]{paillieronline}
Data and Privacy Lab.
\newblock Paillier threshold encryption toolbox, July 2011.

\bibitem[Fei99]{Feige1999}
Uriel Feige.
\newblock Noncryptographic selection protocols.
\newblock In {\em Proceedings of the 40th Annual Symposium on Foundations of
  Computer Science}, FOCS '99, pages 142--, Washington, DC, USA, 1999. IEEE
  Computer Society.

\bibitem[GGHK10]{Giurgiu2010}
A~Giurgiu, R~Guerraoui, K~Huguenin, and A-M Kermarrec.
\newblock {Computing in Social Networks}.
\newblock {\em Lecture Notes in Computer Science}, 6366/2010:332--346, 2010.

\bibitem[GMW87]{Goldreich1987}
O.~Goldreich, S.~Micali, and A.~Wigderson.
\newblock {How to play any mental game}.
\newblock In {\em Proceedings of the nineteenth annual ACM symposium on Theory
  of computing}, pages 218--229. ACM, 1987.

\bibitem[Gol01]{Goldreich2001}
Oded Goldreich.
\newblock {\em {Foundations of cryptography: Basic tools}}.
\newblock Cambridge Univ Press, 2001.

\bibitem[HZ10]{Hirt2010}
Martin Hirt and Vassilis Zikas.
\newblock {Adaptively secure broadcast}.
\newblock {\em Advances in Cryptology–EUROCRYPT 2010}, pages 466--485, 2010.

\bibitem[JCJ10]{Juels2010}
Ari Juels, Dario Catalano, and Markus Jakobsson.
\newblock {Coercion-resistant electronic elections}.
\newblock {\em Towards Trustworthy Elections}, pages 37--63, 2010.

\bibitem[KKK{\etalchar{+}}08]{Kapron2008}
Bruce Kapron, David Kempe, Valerie King, Jared Saia, and Vishal Sanwalani.
\newblock Fast asynchronous byzantine agreement and leader election with full
  information.
\newblock In {\em Proceedings of the nineteenth annual ACM-SIAM symposium on
  Discrete algorithms}, SODA '08, pages 1038--1047, Philadelphia, PA, USA,
  2008. Society for Industrial and Applied Mathematics.

\bibitem[KS10]{King2010}
Valerie King and Jared Saia.
\newblock Breaking the o(n2) bit barrier: Scalable byzantine agreement with an
  adaptive adversary.
\newblock In {\em Proceeding of the 29th ACM SIGACT-SIGOPS symposium on
  Principles of distributed computing}, PODC '10, pages 420--429, New York, NY,
  USA, 2010. ACM.

\bibitem[KSSV06]{King2006}
Valerie King, Jared Saia, Vishal Sanwalani, and Erik Vee.
\newblock Scalable leader election.
\newblock In {\em Proceedings of the seventeenth annual ACM-SIAM symposium on
  Discrete algorithm}, SODA '06, pages 990--999, New York, NY, USA, 2006. ACM.

\bibitem[NS11]{Nishide2011}
Takashi Nishide and Kouichi Sakurai.
\newblock {Distributed paillier cryptosystem without trusted dealer}.
\newblock {\em Information Security Applications}, pages 44--60, 2011.

\bibitem[Pai99]{Paillier1999}
P~Paillier.
\newblock {Public-key cryptosystems based on composite degree residuosity
  classes}.
\newblock {\em Advances in Cryptology EUROCRYPT'99}, 1999.

\bibitem[WLS{\etalchar{+}}02]{White+:osdi02}
Brian White, Jay Lepreau, Leigh Stoller, Robert Ricci, Shashi Guruprasad, Mac
  Newbold, Mike Hibler, Chad Barb, and Abhijeet Joglekar.
\newblock An integrated experimental environment for distributed systems and
  networks.
\newblock In {\em Proc.\ of the Fifth Symposium on Operating Systems Design and
  Implementation}, pages 255--270, Boston, MA, December 2002. {USENIX}
  {Association}.

\bibitem[Yao82]{Yao1982}
AC~Yao.
\newblock {Protocols for secure computations}.
\newblock {\em Proceedings of the 23rd Annual IEEE Symposium on}, pages
  160--164, November 1982.

\bibitem[YHM{\etalchar{+}}09]{Yuen2009}
T.~Yuen, Qiong Huang, Yi~Mu, Willy Susilo, D.~Wong, and G.~Yang.
\newblock {Efficient non-interactive range proof}.
\newblock {\em Computing and Combinatorics}, pages 138--147, 2009.

\bibitem[YKGK10]{Young2010}
Maxwell Young, A.~Kate, Ian Goldberg, and M.~Karsten.
\newblock {Practical robust communication in DHTs tolerating a byzantine
  adversary}.
\newblock In {\em 2010 International Conference on Distributed Computing
  Systems}, pages 263--272. IEEE, 2010.

\end{thebibliography}

\enddocument